\documentclass[runningheads]{llncs}
\usepackage[T1]{fontenc}
\usepackage[utf8]{inputenc}

\usepackage{graphicx}
\usepackage{tabularx}
\usepackage[textsize=footnotesize]{todonotes}
\usepackage{xcolor,colortbl}
\usepackage{hyperref}
\usepackage{amssymb,amsmath}
\usepackage[capitalise, noabbrev]{cleveref}
\usepackage{booktabs}
\usepackage{multirow}
\usepackage{rotating}
\usepackage{cite}
\usepackage{caption}
\usepackage{subcaption}

\usepackage{color}
\usepackage[linewidth=1pt]{mdframed}

\usepackage{scrextend}          %
\usepackage{needspace}          %
\usepackage{bbm}
\usepackage{xcolor,colortbl}

\newcommand{\prob}[6]{%
  \needspace{3\baselineskip}
  \begin{quote}
    \begin{labeling}{#6}%
    \item[#1]
    \item[\emph{#2}]#3
    \item[\emph{#4}]#5
    \end{labeling}%
  \end{quote}%
}

\newcommand{\probdef}[3]{\prob{#1}{Instance:}{#2}{Question:}{#3}{as}}

\newcommand{\N}{\mathbb{N}}
\newcommand{\cons}{\ensuremath{\mathrm{cons1}}}
\newcommand{\splits}{\ensuremath{\mathrm{splits}}}
\newcommand{\conebmin}{\textup{\textsc{Consecutive Block Minimization}}}
\newcommand{\hampath}{\textup{\textsc{Hamiltonian Path}}}

\newcommand{\hamdist}{\ensuremath{d_h}}

\renewcommand{\orcidID}[1]{\href{https://orcid.org/#1}{\includegraphics[scale=.03]{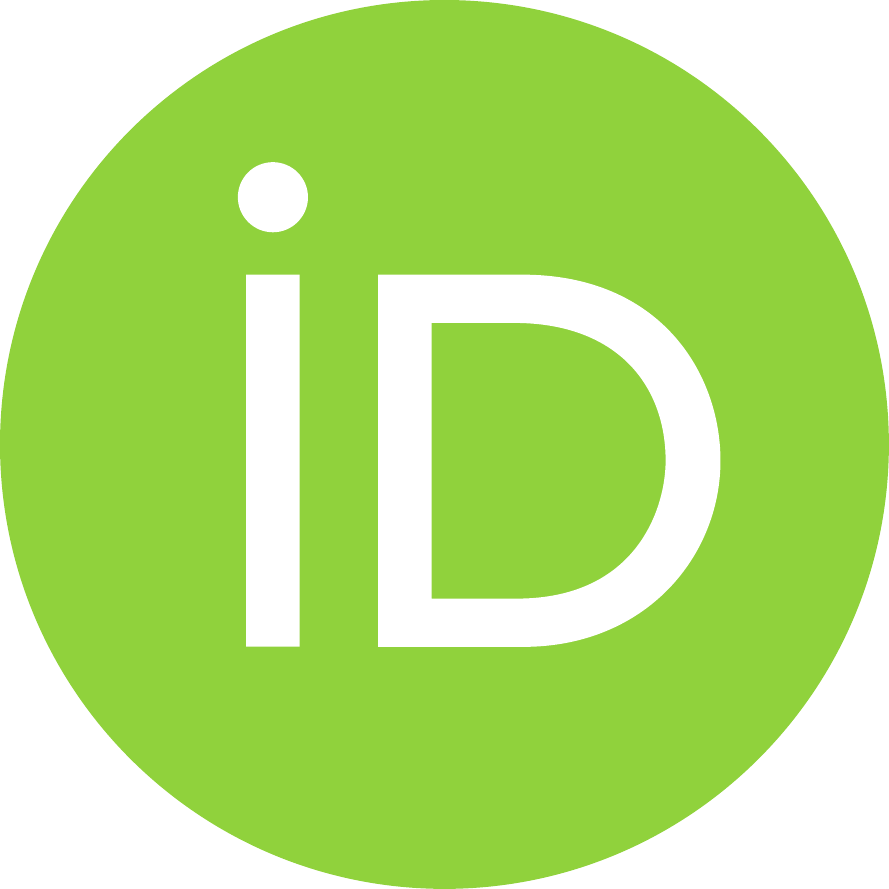}}} 
\usepackage[basic]{complexity}

\begin{document}
\title{On Computing Optimal Linear Diagrams}
\author{Alexander Dobler \orcidID{0000-0002-0712-9726}
\and
  Martin Nöllenburg \orcidID{0000-0003-0454-3937}
}
\authorrunning{A. Dobler and M. Nöllenburg}
\institute{Algorithms and Complexity Group, TU Wien, Vienna, Austria \email{\{adobler,noellenburg\}@ac.tuwien.ac.at}}
\maketitle              %
\begin{abstract}
  Linear diagrams are an effective way to visualize set-based data by representing elements as columns and sets as rows with one or more horizontal line segments, whose vertical overlaps with other rows indicate set intersections and their contained elements. The efficacy of linear diagrams heavily depends on having few line segments. The underlying minimization problem %
  has already been explored heuristically, but its computational complexity has yet to be classified. In this paper, we show that minimizing line segments in linear diagrams is equivalent to a well-studied \NP-hard problem, and extend the \NP-hardness to a restricted setting. %
  We develop new algorithms for computing linear diagrams with minimum number of line segments that build on a traveling salesperson (TSP) formulation %
  and allow constraints on the element orders, namely, forcing two sets to be drawn as single line segments, giving weights to sets, and allowing hierarchical constraints via PQ-trees. We conduct an experimental evaluation and compare previous algorithms for minimizing line segments with our TSP formulation, showing that a state-of-the art TSP-solver can solve all considered instances optimally, most of them within few milliseconds.
\end{abstract}
\section{Introduction}
Many real-world datasets represent set systems, and there is a vast landscape of different visualization techniques for set-based data. Two well-known techniques are Euler and Venn Diagrams that draw sets as closed curves and set intersections are represented by intersections of the boundaries of these curves. For a detailed survey of these and other set visualizations we refer to Alsallakh et al.\ \cite{alsallakhStateoftheArtSetVisualization2016}.

The set visualization that we study in this paper are linear diagrams. 
It has been demonstrated that they are simple and effective, and have advantages when compared with other set visualizations \cite{chapmanVisualizingSetsEmpirical2014,luzComparisonLinearMosaic2019,stapletonEfficacyEulerDiagrams2019}. Linear diagrams represent elements as columns and sets as rows of a matrix or table, where in each row there are one or more horizontal line segments indicating which elements are contained in a specific set. Vertical overlaps of these line segments in different rows show set intersections, and the corresponding elements.
\begin{figure}[t]
  \centering
  \begin{subfigure}{.9\textwidth}
  \includegraphics[width=\textwidth]{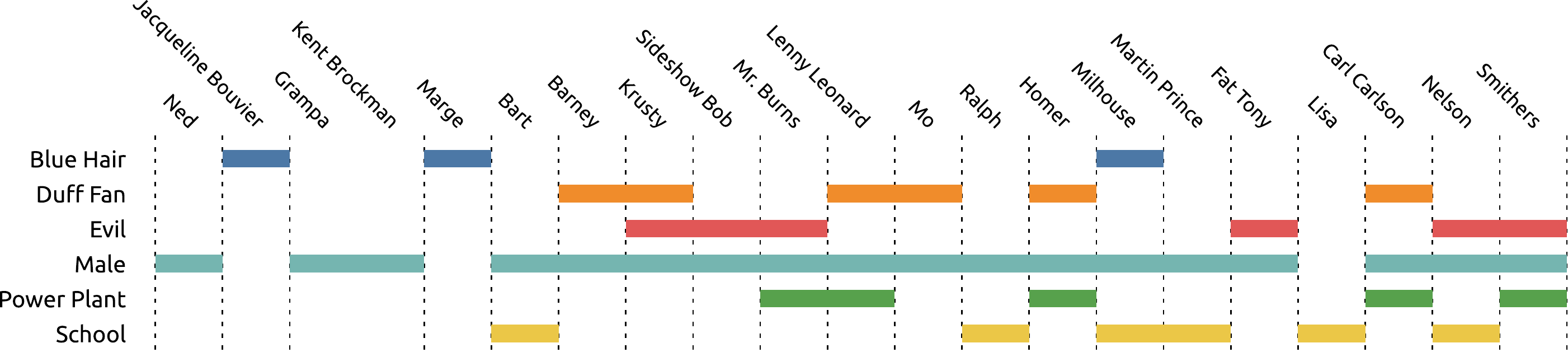}
  \caption{Random ordering of the overlaps.}
  \label{fig:simpsonsbad}
  \end{subfigure}
  \begin{subfigure}{.9\textwidth}
    \includegraphics[width=\textwidth]{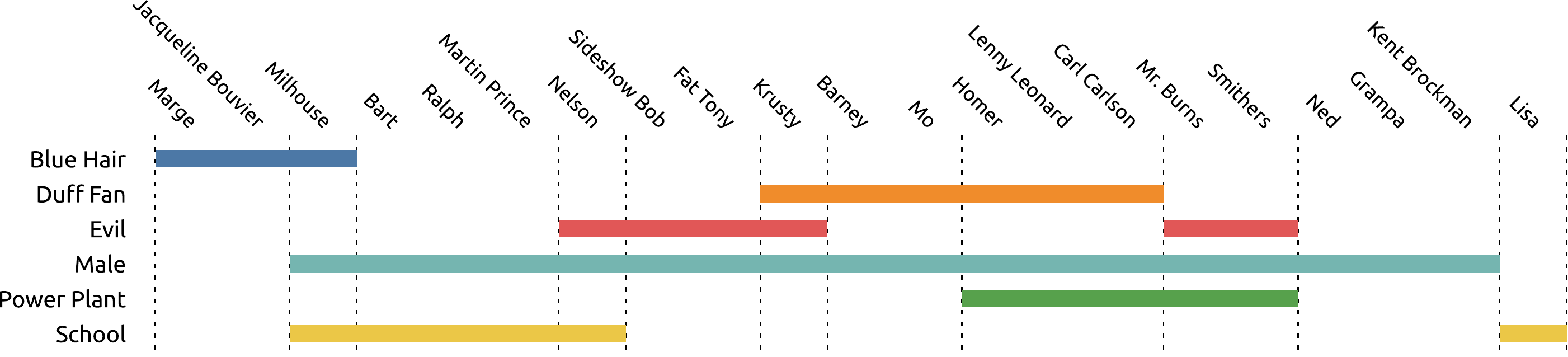}
    \caption{An overlap ordering that minimizes the total number of line segments.}
    \label{fig:simpsonsgood}
  \end{subfigure}
  \caption{Linear diagrams representing the Simpsons.}
  \label{fig:simpsons}
\end{figure}
\cref{fig:simpsonsbad} shows a linear diagram representing a Simpsons data set introduced by Jacobsen et al.~\cite{jacobsenMetroSetsVisualizingSets2021}.
For example, the set Blue Hair contains the elements Jacquelin Bouvier, Marge, and Milhouse, and is drawn with three line segments.
Mr.\ Burns is contained in the sets Evil, Male, and Power Plant, as represented by the corresponding vertical overlap of the line segments in these three rows with the column of Mr.\ Burns.

Linear diagrams can be drawn in many ways, e.g., by choosing different permutations of the rows/sets and columns/overlaps.
It has been shown that there are several quality criteria for linear diagrams, while the most important one is finding an ordering of the elements that minimizes the number of line segments~\cite{rodgersVisualizingSetsLinear2015}. 
For example, the linear diagram depicted in \cref{fig:simpsonsgood} shows the same set system as before, but with an ordering of the overlaps that minimizes the number of line segments, here using 8 segments instead of 23.

The underlying computational problem of finding an ordering of the overlaps that minimizes line segments seems hard, as for $n$ overlaps, there are $n!$ different orderings of these overlaps. Finding orderings that minimize line segments is mainly done via heuristics in the literature~\cite{chapmanDrawingAlgorithmsLinear2021,fedorSupervennPreciseEasytoread2022,rodgersVisualizingSetsLinear2015}. The main topic of this paper is computing \emph{optimal linear diagrams} -- those which realize the minimum possible number of line segments that have to be drawn.

\paragraph*{Related work.}
Several user studies were performed to compare the efficacy of linear diagrams and other diagram types; they showed that linear diagrams perform equally well or better than other diagram types including Euler and Venn diagrams~\cite{chapmanVisualizingSetsEmpirical2014,luzComparisonLinearMosaic2019,satoEfficacyDiagramsSyllogistic2012}.
Linear diagrams have then been used, e.g., to visualize sets over time \cite{masoodianTemporalVisualizationSets2018}, and Lamy et al.\ \cite{lamyRainbowBoxesNew2017} extended linear diagrams to allow multiple sets per row.

Existing algorithms for minimizing line segments in linear diagrams are of heuristic nature, i.e., they may often find good solutions, but do not provide proven guarantees on the solution quality.
Rodgers et al.~\cite{rodgersVisualizingSetsLinear2015} presented a simple heuristic that first defines a pair-wise similarity between two overlaps based on the number of sets they have in common. Then, this heuristic iteratively builds an overlap ordering aiming to group similar overlaps next to each other.
Chapman et al.~\cite{chapmanDrawingAlgorithmsLinear2021} compared different heuristics based on simulated annealing, a travelling salesperson (TSP) formulation, and other variants of the heuristic of Rodgers et al.~\cite{rodgersVisualizingSetsLinear2015}.
A GitHub project \cite{fedorSupervennPreciseEasytoread2022} provides an implementation of linear diagrams in Python. 
The underlying algorithm tries to minimize the number of line segments by applying multiple runs of an iterative greedy heuristic, each with a different pair-wise similarity measure between overlaps that is augmented by random seeds.

\paragraph*{Contribution and structure.}
We further investigate the computational problem of computing optimal linear diagrams.
\cref{sec:prelim} defines general preliminaries and notation for permutations, matrices, and graphs.
In \cref{sec:conblockmin}, we describe how the problem of computing optimal linear diagrams can be modelled as a known problem on binary matrices, thus bridging the gap missing in the literature\footnote{This observation has been made independently in a paper that was published after the first version of our paper \cite{chapmanMinimisingLineSegments2022}.}.
This problem is known to be \NP-complete; we further strengthen this \NP-completeness result by showing that computing optimal linear diagrams is even \NP-complete for set systems where each set contains exactly two elements and each element is contained in exactly three sets. Moreover, we present further literature on matrix problems that are relevant with regard to linear diagrams.

In \cref{sec:TSP}, we present a way to compute optimal linear diagrams by reducing the problem to TSP, thus, completing the work of Chapman et al.~\cite{chapmanDrawingAlgorithmsLinear2021}. They also presented an algorithm based on a TSP formulation, but this algorithm sometimes produces non-optimal overlap orderings. We further expand on this formulation, showing that we can model specific constraints on the overlap orders. Namely, we can force up to two sets to be drawn as single line segments while still minimizing the number of line segments. This is particularly interesting for allowing interactivity in linear diagrams~\cite{chapmanInteractivityLinearDiagrams2021}. 
We also show how to model constraints based on weighted sets and hierarchical ordering constraints represented by PQ-trees, which is of interest for certain set visualization tasks. 

In \cref{sec:experiments}, we conduct an experimental evaluation of our algorithms from \cref{sec:TSP}, and compare them with the state-of-the art heuristics. We show that a state-of-the-art TSP-solver can solve all considered instances optimally, most of them within few milliseconds.
We also verify that the considered heuristics from the literature perform well with regard to the number of line segments, where the average optimality gaps of the heuristics are less than ten percent. 

\section{Preliminaries}\label{sec:prelim}
Let~$A$ be a matrix with~$m$ rows and~$n$ columns; we set $n_A=n$ and $m_A=m$.
We write~$A_{i,j}$ with $1\le i\le m$ and $1\le j\le n$ for the \emph{entry} of~$A$ at row~$i$ and column~$j$.
Furthermore, by~$r^A_i$ and~$c^A_j$ we denote the~$i$-th row and~$j$-th column of~$A$, respectively.
A matrix is a \emph{binary matrix} if all its entries are either 0 or 1.
If it is clear from the context, we might omit explicitly mentioning the matrix~$A$ in
the above notations.

We denote by~$[k]$ the set of elements $\{1,\dots ,k\}$.
A \emph{permutation} $\pi:[k]\to X$ is a bijective function from~$[k]$ to a set~$X$.
Sometimes we write permutations~$\pi$ as sequences of elements, that is, $\pi=(x_1,\dots, x_n)$ is the permutation such that $\pi(i)=x_i$ for $1\le i\le n$.
We denote by~$\Pi_k$ the set of all permutations from~$[k]$ to~$[k]$.
For two permutations $\pi_1=(x_1,\dots,x_n)$ and $\pi_2=(y_1,\dots,y_m)$, we denote by $\pi_1\star \pi_2$ their \emph{concatenation} $(x_1,\dots x_n,y_1,\dots,y_m)$.
For two sets~$\Pi_1$ and~$\Pi_2$ of permutations, we define $\Pi_1\star \Pi_2=\{\pi_1\star \pi_2\mid \pi_1\in\Pi_1,\pi_2\in\Pi_2\}$.

For a matrix~$A$ and a permutation $\pi:[n_A]\to [n_A]$ we denote by $\pi(A)$ the matrix such that $\pi(A)_{i,j}=A_{i,\pi(j)}$. 
Equivalently $\pi(r^A_i)=r^A_{\pi(i)}$ for a row~$r^A_i$. 
By ``a permutation of the columns of the matrix~$A$'' we mean a permutation $\pi:[n_A]\to [n_A]$.

A block of \emph{consecutive ones} in a row~$r^A_i$ of a matrix~$A$ with~$n$ columns is a maximal non-empty sequence $A_{i,p},A_{i,p+1},\dots,A_{i,q}$ satisfying \begin{itemize}
  \item $A_{i,j}=1$ for all $p\le j\le q$,
  \item $p=1$ or $A_{i,p-1}=0$, and
  \item $q=n$ or $A_{i,q+1}=0$.
\end{itemize}  
For a row~$r^A_i$,~$\cons(r^A_i)$ is the number of  blocks of consecutive ones in~$r^A_i$.
Additionally,~$\splits(r^A_i)$ (the number of gaps between the blocks) is defined as $\cons(r^A_i)-1$ if~$r^A_i$ contains a 1-entry, and 0 otherwise.
We define $\cons(A)=\sum_{i=1}^{m_A}\cons(r^A_i)$ for a matrix~$A$.
Equivalently, $\splits(A)=\sum_{i=1}^{m_A}\splits(r^A_i)$.
Let~$c^A_i$ and~$c^A_j$ be two columns of a binary matrix. By~$\hamdist(c^A_i,c^A_j)$ we denote the Hamming distance between~$c^A_i$ and~$c^A_j$, that is, the number of rows with different values.

In this paper we assume graphs $G$ as simple and undirected.
By~$V(G)$ and~$E(G)$ we denote the vertex set and edge set of~$G$, respectively.
For a binary matrix~$A$, let~$G(A)$ be the complete graph that consists of the vertices $V = \{v_i\mid c^A_i\text{ is a column in }A\}$.
If we talk about a vertex~$v_i$ with index~$i$ in~$G(A)$, we mean the vertex~$v_i$ that corresponds to column~$c_i^A$.

Sometimes we consider graphs~$G(A)$ obtained from a matrix~$A$ with a quadratic and symmetric \emph{distance matrix}~$D$ of size $|V(G)|\times |V(G)|$, such that~$D_{i,j}$ is the \emph{length} of the edge between~$v_i$ and~$v_j$.
A \emph{tour}~$T$ in~$G(A)$ is a sequence of vertices $(v_{i_1},\dots,v_{i_{n}})$ that contains each vertex of~$G(A)$ exactly once. (We do not require adjacency, as~$G(A)$ is complete.)
The \emph{length} of~$T$ in~$G(A)$ under a distance matrix~$D$ is $D_{i_n,i_{1}}+\sum_{k=1}^{n-1}D_{i_k,i_{k+1}}$.
Finding a tour of minimum length in~$G(A)$ under a distance matrix~$D$ is known as the \emph{Travelling Salesperson Problem} (TSP) and is \NP-complete \cite{papadimitriouEuclideanTravellingSalesman1977}.

\section{Complexity of Linear Diagrams}\label{sec:conblockmin}
The most important quality aspect supporting the cognitive effectiveness of linear diagrams is the number of line segments \cite{rodgersVisualizingSetsLinear2015}.
To minimize the number of line segments that have to be drawn, we have to find an appropriate horizontal ordering of the overlaps.
There is a one-to-one correspondence between linear diagrams and binary matrices:
Let $(\mathcal{S},\mathcal{U})$ be a set system with universe $\mathcal{U}=\{u_1,\dots, u_n\}$ and sets $\mathcal{S}=\{S_1,\dots,S_m\}$, hence, for all $i\in [m]$, $S_i\subseteq \mathcal{U}$.
The system~$\mathcal{S}$ can be represented by a binary matrix~$A$ s.t.\ $A_{i,j}=1$ if and only if element~$u_j$ belongs to set~$S_i$.
The rows and columns of~$A$ are exactly the rows and columns of the linear diagram, respectively.
Line segments in the linear diagram correspond to blocks of consecutive ones in the matrix~$A$.
The problem of finding a horizontal ordering of the overlaps that minimizes the number of line segments is equivalent to the problem of finding a permutation $\pi\in\Pi_n$ that minimizes~$\cons(\pi(A))$.

A matrix~$A$ is said to have the \emph{consecutive ones property} (C1P) if there is a permutation $\pi\in\Pi_{n_A}$ with $\splits(\pi(A))=0$.
There are several linear-time algorithms for testing if a matrix has the C1P and for computing the corresponding permutation, the first due to Booth and Lueker \cite{boothTestingConsecutiveOnes1976}.
Thus, we can decide in linear time if a linear diagram can be drawn such that each set is represented by exactly one line segment.

Most of the time though, linear diagrams cannot be drawn in this way.
In this case we want to minimize the number of required line segments.
The corresponding binary matrix problem is known as \emph{consecutive block minimization} in the literature, its decision problem is given below.
\probdef{\conebmin}
{A binary matrix~$A$ and a non-negative integer~$k$.}
{Does there exist a permutation $\pi\in \Pi_{n_A}$ such that \\$\cons(\pi(A))\le k$?}
The problem has been shown to be \NP-complete \cite{kouPolynomialCompleteConsecutive1977}, even if each row contains exactly two ones \cite{haddadiNoteNPHardness2002}.
We give here an alternative proof of \NP-completeness for binary matrices with two ones per row and three ones per column, thus further strengthening the \NP-completeness result.
\begin{theorem}
  \conebmin\ is \NP-complete for matrices with two ones per row and three ones per column.
\end{theorem}
\begin{proof}
  Membership in \NP\ is evident.
  For hardness, we give a reduction from \hampath\ on graphs of degree 3, which is \NP-complete \cite{gareyPlanarHamiltonianCircuit1976}.
  \hampath\ asks for a given graph~$G$, if there is a path in~$G$ that visits every vertex exactly once.
  Let~$G$ be an instance of \hampath\ such that $E(G)=\{e_1,\dots, e_m\}$ and $V(G)=\{v_1,\dots, v_n\}$ and~$G$ has degree 3.
  We construct an instance~$(A,k)$ of \conebmin\ as follows.
  Let~$A$ be the incidence matrix of $G$, which has $n_A=|V(G)|$ columns and $m_A=|E(G)|$ rows with $A_{i,j}=1$ if and only if $v_i\in e_j$.
  Clearly, this matrix has two ones per row, as each edge contains two vertices and 3 ones per column, as $G$ has degree 3.
  We show that~$G$ contains a Hamiltonian path if and only if there exists a permutation~$\pi$ of the columns of~$A$ such that $\cons(\pi(A))\le 2\cdot m-(n-1)$.
  
  ``$\Rightarrow$'': Let $P=(v_{\ell_1},v_{\ell_2}, \dots , v_{\ell_n})$ be a Hamiltonian path in~$G$.
  We claim that $\pi=(\ell_1,\ell_2,\dots, \ell_n)$ satisfies $\cons(\pi(A))\le 2\cdot m-(n-1)$.
  Consider the edges $\{v_{\ell_i},v_{\ell_{i+1}}\}$ for $1\le i \le n-1$, which exist because~$P$ is a path.
  As~$v_{\ell_i}$ and~$v_{\ell_{i+1}}$ are consecutive in~$P$, the columns~$c^A_{\ell_i}$ and~$c^A_{\ell_{i+1}}$ are consecutive in~$\pi(A)$.
  Thus, the row in~$A$ corresponding to the edge $\{v_{\ell_i},v_{\ell_{i+1}}\}$ contributes to exactly one block of consecutive ones.
  The remaining $m-(n-1)$ rows can contribute to at most two blocks of consecutive ones as they only contain two 1-entries each.
  Together, there are at most $n-1+2\cdot(m-(n-1))=2\cdot m-(n-1)$ blocks of consecutive ones in~$\pi(A)$.
  
  ``$\Leftarrow$'': Let $\pi=(\ell_1,\ell_2,\dots ,\ell_n)$ be a permutation of the columns of~$A$ that satisfies $\cons(\pi(A))\le 2\cdot m-(n-1)$.
  We claim that $P=(v_{\ell_1},v_{\ell_2}, \dots , v_{\ell_n})$ is a Hamiltonian path in~$G$.
  There are at least $n-1$ blocks of consecutive ones of size two in~$\pi(A)$ as otherwise $\cons(\pi(A))> 2\cdot m-(n-1)$.
  As $G$ is a simple graph, no two rows of~$A$ contain ones in the same columns and thus each of these blocks of consecutive ones has to start at a different column.
  By the pigeonhole principle, for each $1\le i\le n-1$, there exists such a block of consecutive ones that starts at the~$i$-th column of~$\pi(A)$.
  Hence, $\{v_{\ell_i},v_{\ell_{i+1}}\}$ is an edge in~$G$ for all $1\le i\le n-1$, and~$P$ is a Hamiltonian path.  \qed
\end{proof}
We have been made aware of the fact that this result has been proved independently previously, using the same reduction \cite{haddadiIteratedLocalSearch2021}.

\conebmin\ has been further studied from an algorithmic view.
Several heuristic methods for finding permutations~$\pi$ with small $\cons(\pi(A))$ have been given \cite{haddadiExponentialNeighborhoodSearch2021,haddadiIteratedLocalSearch2021,soaresHeuristicMethodsConsecutive2020}.
Haddadi and Layouni \cite{haddadiConsecutiveBlockMinimization2008} transformed \conebmin\ to a travelling salesperson problem, we will go into more details on their results in \cref{sec:TSP}.

Further variations of consecutive-ones problems that could be interesting for linear diagrams have been studied, mostly giving hardness results or polynomial algorithms assuming that some underlying parameters of the problems are constant:
It has been shown that the problem of finding a permutation~$\pi$ of the columns of a binary matrix~$A$ such that for all $i\in [m_A]$, $\cons(r_i^A)\le k\in \mathbb{N}$ is \NP-complete~\cite{goldbergFourStrikesPhysical1995}, which translates to the problem of having at most $k$ line segments per set in a linear diagram.
Another more involved problem has been studied, called $\textsc{Gapped Consecutive Ones}$, in which we are given a binary matrix~$A$ and want to find a permutation~$\pi$ of the columns of~$A$ such that for all $i\in [m_A]$, $\cons(r_i^A)\le k\in \mathbb{N}$, and the gaps between two consecutive blocks of ones in a row of~$\pi(A)$ is at most some maximum gap parameter~$\delta$ \cite{chauveGappedConsecutiveOnesProperty2009,manuchHardnessResultsGapped2012,manuchComplexityGappedConsecutiveOnes2011}.
Here gaps refer to maximal blocks of zeros between two blocks of ones.

Furthermore, there is literature devoted to turning a binary matrix into a binary matrix that has the C1P by deleting rows, deleting columns, or flipping entries (turning 1-entries into 0-entries and/or turning 0-entries into 1-entries).
Dom et al.\ \cite{domApproximationFixedparameterAlgorithms2010} give a summary of results.

\conebmin\ is also related to a variant of database compression. Given a database consisting of $m$ rows and $n$ columns, one wants to permute the rows of the database such that the number of \emph{runs} is minimized. Runs are essentially consecutive elements in a column that have the same value. When compressing the database, such runs can be saved using constant memory by only saving the start, end, and value of the run. The main difference to our application is that databases are normally comprised of a huge number of rows, sometimes in the millions. Thus, mainly heuristics have been studied, in particular, heuristics that take less than quadratic time in the number of rows. In our context of linear diagrams, algorithms and heuristics consuming quadratic time in the input are no problem. Some algorithms for minimizing runs in databases are for example described in \cite{lemireReorderingRowsBetter2012}.
\section{TSP Model}\label{sec:TSP}
In this section, we describe the procedure of minimizing the number of line segments in a linear diagram by using a TSP model, and give a runtime optimization.
We also show how to incorporate further constraints into this model.

\subsection{Solving Linear Diagrams with TSP}\label{sec:lindiagTSP}
We now present how to solve the task of minimizing the number of line segments drawn in a linear diagram.
Let us start with the key lemma for our model.

\begin{lemma}[\hspace{-0.01pt}\cite{haddadiConsecutiveBlockMinimization2008}]\label{lemma:tour}\sloppy
  Let~$A$ be a binary matrix with $n$ columns and let~$A^\prime$ be the binary matrix obtained from~$A$ by appending a column of zeros to the right of~$A$.
  Let $(v_{i_1},v_{i_2},\dots,v_{i_{n+1}})$ be a tour of length~$L$ in~$G(A^\prime)$ under distance matrix $D_{i,j}=\hamdist(c^{A^\prime}_i,c^{A^\prime}_j)$.
  Assume that $v_{i_k}=v_{n+1}$, corresponding to the appended column of zeros, and let $\pi=(i_{k+1},\dots i_{n+1},i_{1},\dots, i_{k-1})$.
  Then $L=2\cdot \cons(\pi(A))$.
\end{lemma}

As discussed in \cref{sec:conblockmin}, the task of minimizing line segments in linear diagrams is the same as finding a permutation~$\pi$ of the columns of a binary matrix~$A$ that minimizes~$\cons(\pi(A))$.
One way to find such a permutation is with a TSP-model as outlined by \cref{lemma:tour}:
Let $A$ be a binary matrix with $n$ columns.
We construct the binary matrix~$A^\prime$ by appending a column of zeros to the right of~$A$. %
From the matrix~$A^\prime$, we construct the complete graph~$G(A^\prime)$, such that vertices correspond to columns in~$A^\prime$.
A distance matrix~$D$ for~$G(A^\prime)$ is constructed such that~$D_{i,j}$ is the Hamming distance $\hamdist(c^{A^\prime}_i,c^{A^\prime}_j)$.
We then compute a TSP tour $(v_{i_1},v_{i_2},\dots,v_{i_{n+1}})$ of minimal length in~$G(A^\prime)$.
Assume that~$v_{i_k}$ is the vertex corresponding to the column $c^{A^\prime}_{n_{A^\prime}}$.
Then, by \cref{lemma:tour}, $\pi=(i_{k+1},\dots,i_{n+1},i_1,\dots,i_{k-1})$ is the permutation with minimal~$\cons(\pi(A))$.
The intuition for this is that choosing an edge~$\{v_i,v_j\}$ of small length in~$G(A^\prime)$ is the same as starting or ending few consecutive blocks of ones (corresponding to line segments in a linear diagram) when going from the column~$c_i$ to~$c_j$.
With this argumentation each block of consecutive ones is started and ended exactly once, and the length of the tour is $2\cdot\cons(\pi(A))$.
Note that adding the extra column at the end is necessary, as otherwise it could be that some consecutive blocks of ones, those that start at the first column or end at the last column, are ``not counted in the tour''.

There is a small runtime optimization that can be applied to decrease the size of the graph~$G(A^\prime)$.
Columns of~$A$ that have ones in the same rows, their Hamming distance being zero, can be collapsed into a single column.
The above procedure may be applied to compute the desired permutation of columns, and then the collapsed columns can be expanded again to appear consecutively.
Clearly, this does not influence the number consecutive blocks of ones in the resulting matrix. In terms of set systems, this corresponds to collapsing multiple overlaps that contain the same sets into a single representative. In an optimal linear diagram, such overlaps would never be separated.

We tested this method of computing optimal column orderings by applying a state-of-the art TSP-solver.
We will report on experimental results for real-world and previously considered set visualization instances in \cref{sec:experiments}.
Note that the same procedure has already been applied to instances from consecutive block minimization \cite{soaresHeuristicMethodsConsecutive2020}.

\subsection{Priorities for Sets}\label{sec:prioritiessets}
In some contexts certain sets in a linear diagram might be considered more important than others.
We would want to compute a linear diagram, in which these sets are drawn with a single line segment, but  the other sets should be drawn with as few line segments as possible.
It is clear that forcing more than two sets to be drawn as one line segment is not always possible, as there are binary matrices with three rows that do not have the C1P.
We can solve the problem on binary matrices as a TSP model due to the following result.
\begin{lemma}\label{lemma:hierarchical}
  Let~$A$ be a binary matrix with $n$ columns and exactly~$p$ 1-entries and let $C_1,\dots,C_q\subseteq \{c_1,\dots,c_{n_A}\}$ be a family of non-empty sets of columns of~$A$ satisfying %
  \[\exists \pi\in\Pi_n\forall k\in[q]:\emph{the columns in }C_k\text{ appear consecutively in }\pi(A).\]
  Let~$A^\prime$ be the matrix obtained from~$A$ by appending a column of zeros.
  We consider the graph~$G(A^\prime)$ with distance matrix $D$ s.t.
  \[D_{i,j}=\hamdist(c_i,c_j)+(2p+1)\cdot \sum_{k=1}^q|\mathbbm{1}_{C_k}(c_i)-\mathbbm{1}_{C_k}(c_j)|,\]
  where~$\mathbbm{1}_{C_k}$ is the indicator function for set~$C_k$.
  Let $T=(v_{i_1},v_{i_2},\dots,v_{i_{n+1}})$ be a tour of minimal length in~$G(A^\prime)$ under distance matrix $D$.
  Let~$v_{i_k}=v_{n+1}$, corresponding to the appended column of zeros.
  Then the permutation $\pi=(i_{k+1}, \dots, i_{n+1},i_1,\dots, i_{k-1})$ has the following properties
  \begin{enumerate}
    \item[(1)] For all $k\in [q]$ the columns in~$C_k$ appear consecutively in~$\pi(A)$.
    \item[(2)] Of all $\pi^\prime\in\Pi_n$ that satisfy (1),~$\pi$ is the one with minimum~$\cons(\pi(A))$.
  \end{enumerate}
\end{lemma}
\begin{proof}
  Let~$\pi$ be the permutation as defined above.
  We first show by contradiction that~$\pi$ satisfies (1).
  Assume to the contrary that~$\pi$ does not satisfy (1) and consider any permutation~$\pi^\prime$ of the columns of~$A$ that satisfies (1). This permutation exists by assumption.
  Consider the tour $T^\prime=(v_{n+1})\star\pi$.
  The length of~$T^\prime$ is at most $2q(2p+1)+2p$, as there can be at most~$p$ consecutive blocks of ones in~$\pi^\prime(A)$, each contributing two to the length of~$T^\prime$.
  The value $2q(2p+1)$ is due to the fact that we ``leave'' or ``enter'' vertices corresponding to the set of columns~$C_k$,~$1\le k\le q$, exactly twice.
  To the contrary, the length of~$T$ is at least $2(q+1)(2p+1)$.
  Hence,~$T$ cannot be a tour of minimal length, yielding a contradiction.
  It is clear that~$\pi$ also satisfies (2), as increases in the length of the tour~$T$, also increases the number of consecutive ones of the matrix~$\pi(A)$ due to the same reasoning as in \cref{lemma:tour}. \qed
\end{proof}
We can directly apply the above lemma to find a permutation of the columns of a matrix~$A$ with minimum blocks of consecutive ones among the permutations~$\pi$ that have $\cons(r_{i_1}^{\pi(A)}))=\cons(r_{i_1}^{\pi(A)})=1$ for $i_1,i_2\in m_A$:
We simply define $C_1=\{j\in [n_A]\mid A_{i_1,j}=1\}$ and $C_2=\{j\in [n_A]\mid A_{i_2,j}=1\}$, and apply the reduction to TSP as outlined in \cref{lemma:hierarchical}.
Clearly, $C_1$ and $C_2$ satisfy the requirements of \cref{lemma:hierarchical}, as a matrix with two rows always has the C1P.
In our experiments we show how adding these constraints affects the runtime and the number of blocks of consecutive ones.
Note, however, that the result of \cref{lemma:hierarchical} allows us to constrain column orders of a matrix in far more general ways.

\subsection{A Weighted Version}\label{sec:constraintsweight}
The shortcoming of the approach described in \cref{sec:prioritiessets} is that we can only restrict two sets to be drawn in one line segment.
If we want to involve more sets in this, we can use a model with weighted sets (corresponding to rows in a binary matrix~$A$).
Then, if a weight of row~$r_i^A$ is bigger than a weight of~$r_j^A$, it is ``worse'' to have more blocks of consecutive ones for~$r_i^A$ than it is for~$r_j^A$.
Formally, we are given a binary matrix~$A$ and a weight function $f:[m_A]\to \mathbb{N}$, and we want to find a permutation~$\pi$ of the columns of~$A$ that minimizes $\sum_{i=1}^{m_A}f(i)\cons(r_i^{\pi(A)})$.
Solving this problem is straight-forward with a TSP-model:
We construct the matrix~$A^\prime$ by appending a column of zeros to the right of~$A$.
We then create a distance matrix $D$ for~$G(A^\prime)$ such that $D_{i,j}=\sum_{k=1}^{m_A}f(k)|A^\prime_{k,i}-A^\prime_{k,j}|$.
This distance matrix corresponds to weighted Hamming distances.
Then, we simply find the tour~$T$ of minimal total distance in~$G(A^\prime)$ under $D$, and obtain the desired permutation~$\pi$ from~$T$ as in \cref{lemma:tour}.
We conduct experiments for this weighted version of consecutive block minimization in \cref{sec:experiments}.

\subsection{Hierarchical Constraints}
We now to present an algorithm that allows for more general constraints on the allowed column orders of a binary matrix, restricting column orders by \emph{PQ-trees}.
We adopt the definition of PQ-trees of Burkard et al.\ \cite{burkardTravellingSalesmanPQtree1996}, as our algorithm directly applies their results.
A PQ-tree~$T$ over the set~$[n]$ is a rooted, ordered tree whose \emph{leaves} are pairwise distinct elements of~$[n]$ and whose \emph{internal} nodes are distinguished as either~$P$-nodes or~$Q$-nodes.
The set~$\textsc{leaf}(T)$ denotes the leaves of~$T$.

Every PQ-tree~$T$ represents a set~$\Pi(T)$ of permutations of~$\textsc{leaf}(T)$ as follows.
If~$T$ consists of a single leaf $i\in [n]$, then $\Pi(T)=\{(i)\}$.
Otherwise, the root~$r(T)$ of~$T$ is a~$P$-node or a~$Q$-node.
Let $v_1,\dots,v_m$ denote the children of~$r(T)$, ordered from left to right, and let~$T_i$ denote the maximal subtrees rooted at~$v_i$, $1\le i\le m$.
If~$r(T)$ is a~$P$-node, then
\[\Pi(T)=\bigcup_{\psi\in \Pi_m}\Pi(T_{\psi(1)})\star\Pi(T_{\psi(2)})\star\dots \star \Pi(T_{\psi(m)}),\]
and if~$r(T)$ is a~$Q$-node, then
\[\Pi(T)=\Pi(T_1)\star \Pi(T_2)\star\dots\star \Pi(T_m)\cup \Pi(T_m)\star \Pi(T_{m-1})\star\dots\star\Pi(T_1).\]
Informally, children of~$P$-nodes can be permuted arbitrarily, while children of~$Q$-nodes can only be reversed.

For applications of PQ-trees we refer to Booth and Lueker~\cite{boothTestingConsecutiveOnes1976}.
They can be used to model allowed column orders of a binary matrix or, equivalently, the orders of overlaps in a linear diagram; for example, if overlaps illustrated by a linear diagram have some hierarchical relations between them and should not be permuted arbitrarily, then we might represent this by a PQ-tree accordingly.
If the maximum degree of the PQ-tree~$T$ that represents column orders of a binary matrix~$A$ has small maximum degree, then the permutation $\pi\in\Pi(T)$ that minimizes $\cons(\pi(A))$ can be computed efficiently.
\begin{figure}[!tb]
  \centering
  \includegraphics{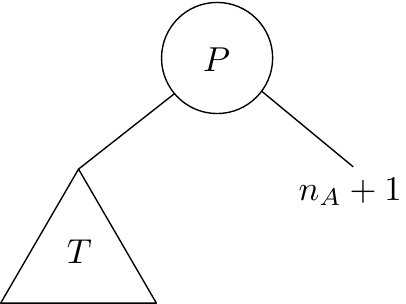}
  \caption{Construction of PQ-tree in \cref{thm:PQtree}.}
  \label{fig:pqtree}
\end{figure}
\begin{theorem}\label{thm:PQtree}
  Let~$A$ be a binary matrix and let~$T$ be a PQ-tree of maximum degree~$d$ such that $\Pi(T)\subseteq \Pi_{n_A}$.
  The permutation $\pi\in \Pi(T)$ that minimizes $\cons(\pi(A))$ can be found in time $\mathcal{O}(\max(m_A\cdot n_A^2,2^d\cdot n_A^3))$.
\end{theorem}
\begin{proof}
  We apply a result of Burkard et al.~\cite{burkardTravellingSalesmanPQtree1996} that states that for a PQ-tree~$T$ with maximum degree~$d$, and an $n \times n$ distance matrix~$D$, the shortest TSP tour for the matrix~$D$ contained in~$\Pi(T)$ can be computed in $\mathcal{O}(2^d\cdot n^3)$ overall time.
  
  Let~$A$ be a binary matrix and let~$T$ be a PQ-tree of maximum degree~$d$ such that $\Pi(T)\subseteq\Pi_{n_A}$.
  Let~$A^\prime$ be the binary matrix obtained from~$A$ by appending a column of zeros to the right of~$A$.
  We construct a PQ-tree~$T^\prime$ such that $\Pi(T^\prime)\subseteq \Pi_{n_A+1}$.
  The PQ-tree~$T^\prime$ consists of a~$P$-node that has two children: The leaf $n_A+1$ and the tree~$T$ rooted at~$r(T)$ (see \cref{fig:pqtree}).
  Notice that the maximum degree of~$T^\prime$ is at most $d+1$.
  Let~$D$ be the distance matrix corresponding to edge weights $D_{i,j}=\hamdist(c^{A^\prime}_i,c^{A^\prime}_j)$ in~$G(A^\prime)$.
  Due to the result of Burkard et al. we can find in time $\mathcal{O}(2^d\cdot n_A^3)$ a tour of minimum length in $G(A^\prime)$ that is contained in~$\Pi(T^\prime)$.
  By \cref{lemma:tour} and the construction of~$T^\prime$, we can obtain from this tour a permutation $\pi\in\Pi(T)$ that minimizes~$\cons(\pi(A))$. We need $m_A\cdot n_A^2$ time to construct the distance matrix $D$, thus we need to account for the possibility that $m_A\cdot n_A^2>2^d\cdot n_A^3$, taking the maximum of both. \qed
\end{proof}

\section{Experiments}\label{sec:experiments}
In this section, we present an experimental evaluation of the algorithms proposed in \cref{sec:TSP}, comparing them with state-of-the art heuristics.
\subsection{Setup and Test Data}
\paragraph*{Setup.}
All experiments were performed on a desktop machine with an Intel i7-8700K processor.
The implementations of algorithms were done in Python 3.7.
To solve our TSP models, we used the Concorde TSP solver\footnote{\url{https://www.math.uwaterloo.ca/tsp/concorde.html}} with the QSopt linear programming solver\footnote{\url{http://www.math.uwaterloo.ca/~bico/qsopt/}}. The code is available online~\cite{doblalex_2022}.
\paragraph*{Test data.}
We consider binary matrices from two different sources.
The first set of instances, referred to as~$T_1$, is taken from Chapman et al.\ \cite{chapmanDrawingAlgorithmsLinear2021} and is available online\footnote{\url{https://doi.org/10.17869/enu.2021.2748170}}.
These instances consist of 440 binary matrices with 5 sets and 10 overlaps up to 50 sets and 70 overlaps.
Chapman et al.~\cite{chapmanDrawingAlgorithmsLinear2021} provide results of their algorithms for minimizing line segments for these instances.

The second set of instances, referred to as~$T_2$, comes from a work by Jacobsen et al.\ \cite{jacobsenMetroSetsVisualizingSets2021} and is available online\footnote{\url{https://osf.io/nvd8e/}}.
The set systems represented by these instances are taken out from a large real-world dataset coming from the Kaggle ``What's Cooking'' competition \cite{amburgClusteringGraphsHypergraphs2020a}.
The sizes of these instances range from 20 overlaps and 6 sets to 160 overlaps and 20 sets.
Overall, there are a total of 4060 instances.

\subsection{Computing Optimal Linear Diagrams}
The first set of experiments considers the task of computing optimal linear diagrams, or equivalently, finding column orderings of the instances that minimize the number of blocks of consecutive ones.

\paragraph*{Algorithms.}
We include comparisons of the following algorithms.
\begin{itemize}
  \item TSPConcorde: This algorithm from \cref{sec:lindiagTSP}uses our TSP   model and the Concorde TSP solver to solve the problem optimally.
        The reported runtimes include generating input files for the Concorde solver and reading its output.
  \item HeuristicRodgers: This algorithm is a python implementation of a greedy algorithm by Rodgers et al.\ \cite{rodgersVisualizingSetsLinear2015}.
        A pairwise similarity measure between overlaps is defined, and then an overlap order is computed iteratively, trying to place similar overlaps next to each other.
        Rodgers et al.\ provide an online demo that implements this algorithm\footnote{\url{http://www.eulerdiagrams.com/linear/generator/}}.
  \item Supervenn: This algorithm is from a recent GitHub project \cite{fedorSupervennPreciseEasytoread2022}.
        For a set of 10000 seeds it defines a pairwise similarity measure between overlaps and then applies a heuristic to compute an overlap order.
  \item BestChapman: Chapman et al.~\cite{chapmanDrawingAlgorithmsLinear2021} compare several heuristic methods to compute overlap orderings of linear diagrams that minimize the drawn line segments.
        They report the number of line segments of overlap orders computed by their algorithms for test set $T_1$.
        As they do not provide the code for all algorithms, and the explanation of the remaining algorithms is incomplete, we had to restrict the evaluation of their approaches to test set $T_1$.
        For an instance of $T_1$, we assume that the algorithm BestChapman is any algorithm of Chapman et al.\ that computes an overlap ordering with the least amount of blocks of consecutive ones.
        They do not provide the runtimes of their algorithms in their abstract~\cite{chapmanDrawingAlgorithmsLinear2021}, so we cannot either. 
\end{itemize}

\paragraph*{Comparison.}
TSPConcorde by design computes optimal column orderings. Hence, we report the relative and absolute optimality gaps for the other algorithms.
That is, let $\text{blocks}(\mathcal{A},I)$ be the number of blocks of consecutive ones of a column ordering computed by algorithm $\mathcal{A}$ for instance $I$.
Then the relative optimality gap in percent is $100\cdot (\frac{\text{blocks}(\mathcal{A},I)}{\text{blocks}(\text{TSPConcorde},I)}-1)$ and the absolute optimality gap is $\text{blocks}(\mathcal{A},I)-\text{blocks}(\text{TSPConcorde},I)$.
For a set of instances, we report these value averaged.
For TSPConcorde we provide the average number of consecutive blocks of ones per row, as the optimality gap is always zero.
We also provide the mean runtime for the same set of instances.
Results are broken down by the number of columns, as the factorial of the number of columns determines the size of the possible search space for an algorithm.  

\paragraph*{Test set $T_1$.}
\begin{table}[!t]
  \centering
  \caption{Results for test set $T_1$. White columns depict the mean relative/absolute optimality gaps (except TSPConcorde); gray columns depict the mean runtimes. For the algorithms of Chapman we do not know the runtimes.}
  \label{table:expt1}
  \begin{tabularx}{\linewidth}{l >{\centering\arraybackslash}X >{\centering\arraybackslash\columncolor{lightgray!15}}X >{\centering\arraybackslash}X >{\centering\arraybackslash\columncolor{lightgray!15}}X >{\centering\arraybackslash}X >{\centering\arraybackslash\columncolor{lightgray!15}}X >{\centering\arraybackslash}X }
    \toprule
    \multirow{4}*{\#columns} & \multicolumn{2}{c}{TSPConcorde} & \multicolumn{2}{c}{HeuristicRodgers} & \multicolumn{2}{c}{Supervenn} & Best Chapman                                    \\
                             & blocks /                        & $t$                                  & gap                           & $t$          & gap         & $t$  & gap         \\
                             & row                             & [ms]                                 & [rel./abs.]                   & [ms]         & [rel./abs.] & [ms] & [rel./abs.] \\\midrule
    10                       & 1.7                             & 7                                    & 3.3/0.9                       & 0            & 1.6/0.8     & 853  & 0.0/0.0     \\
    20                       & 2.8                             & 13                                   & 6.0/3.0                       & 1            & 3.3/1.8     & 1360 & 0.0/0.0     \\
    30                       & 4.0                             & 22                                   & 6.0/5.2                       & 1            & 3.4/3.1     & 1861 & 0.2/0.3     \\
    50                       & 6.0                             & 69                                   & 6.9/9.3                       & 4            & 3.7/5.6     & 2969 & 0.4/0.5     \\
    70                       & 7.9                             & 340                                  & 8.1/13.3                      & 7            & 4.7/8.0     & 4192 & 0.5/0.8     \\
    \bottomrule
  \end{tabularx}
\end{table}
\cref{table:expt1} shows the results for test set $T_1$.
The simple heuristic of Rodgers et al.~\cite{rodgersVisualizingSetsLinear2015} has the smallest runtimes, while also performing worst with regard to optimality gaps.
The runtimes of Supervenn are rather high, while the optimality gaps are lower when compared to HeuristicRodgers, resulting from the 10000 runs of a heuristic, each skewed with a different seed value. 
While the problem of consecutive block minimization is \NP-complete, TSPConcorde solved all instances optimally. 
The average runtime for the largest class of instances from~$T_1$ is still less than a second.
It is worth mentioning that optimality gaps of mostly under 10\% indicate that the heuristics are quite good.

The heuristics of Chapman et al.~\cite{chapmanDrawingAlgorithmsLinear2021} solved 340 of the 440 instances optimally.
For the remaining instances, the maximum difference between the optimal number of consecutive blocks and their best solution is 3.
This yields the fairly small optimality gaps for BestChapman, while we expect that these values would increase for larger instances, a pattern that just starts to appear in \cref{table:expt1}.

\paragraph*{Test set $T_2$.}
\begin{table}[!t]
  \centering
  \caption{Results for test set $T_2$. White columns depict the mean relative/absolute optimality gaps (except TSPConcorde); gray columns depict the mean runtimes.}
  \label{table:expt2}
  \begin{tabularx}{\linewidth}{l >{\centering\arraybackslash}X >{\centering\arraybackslash\columncolor{lightgray!15}}X >{\centering\arraybackslash}X >{\centering\arraybackslash\columncolor{lightgray!15}}X >{\centering\arraybackslash}X >{\centering\arraybackslash\columncolor{lightgray!15}}X}
    \toprule
    \multirow{3}*{\#columns} & \multicolumn{2}{c}{TSPConcorde} & \multicolumn{2}{c}{HeuristicRodgers} & \multicolumn{2}{c}{Supervenn}                              \\
                             & blocks /                        & $t$                                  & gap                           & $t$  & gap         & $t$   \\
                             & row                             & [ms]                                 & [rel./abs.]                   & [ms] & [rel./abs.] & [ms]  \\\midrule
    20-50                    & 1.7                             & 17                                   & 8.4/2.0                       & 1    & 7.7/1.8     & 1949  \\
    55-80                    & 1.8                             & 23                                   & 10.0/2.5                      & 2    & 8.4/2.2     & 3642  \\
    85-110                   & 1.9                             & 36                                   & 10.9/2.8                      & 4    & 9.2/2.5     & 5408  \\
    115-140                  & 2.0                             & 82                                   & 11.1/3.1                      & 6    & 9.8/2.8     & 7782  \\
    145-160                  & 2.0                             & 71                                   & 10.7/3.0                      & 9    & 9.8/2.9     & 10133 \\
    \bottomrule
  \end{tabularx}
\end{table}
\begin{figure}[!t]
  \centering
  \includegraphics[width=.71\textwidth]{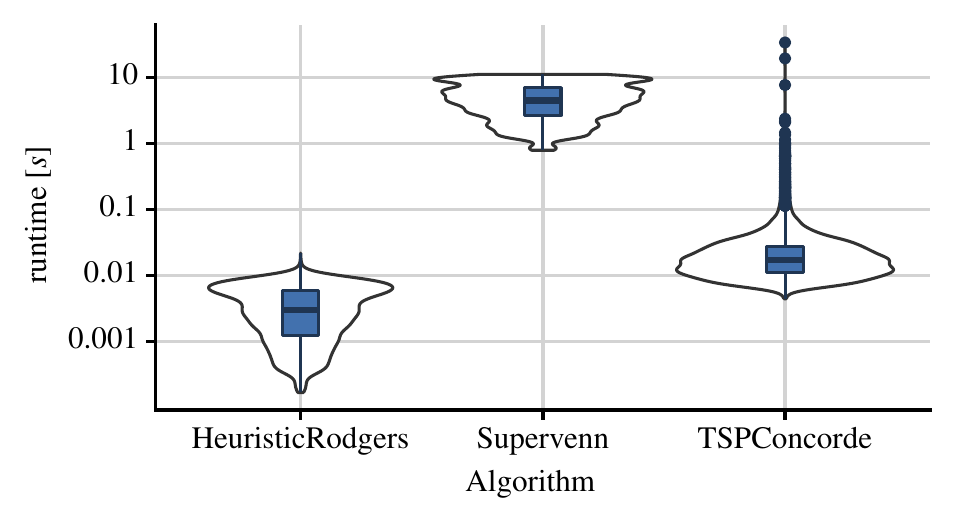}
  \caption{Violin and box plot showing runtimes for all instances from $T_1$ and $T_2$.}
  \label{fig:runtimes}
\end{figure}
\cref{table:expt2} shows results for test set $T_2$.
TSPConcorde is able to solve all instances optimally, the mean runtime still being well below 100ms, even for instances with up to 160 columns.
For Supervenn and HeuristicRodgers we see similar results as in the previous test set.
While Supervenn has slightly better optimality gaps, HeuristicRodgers takes only a thousandth of the time of Supervenn.
Again, optimality gaps increase with increasing number of columns to about 10 percent compared to the optimal solutions. 

\paragraph*{Runtimes.} \cref{fig:runtimes} shows a boxplot and violin plot of the runtimes of the three algorithms HeuristicRodgers, Supervenn, and TSPConcorde for the combined test set $T_1\cup T_2$.
The $y$-axis is scaled logarithmically.
It again reflects that Supervenn takes much longer than HeuristicRodgers, while the runtimes for both algorithms do not contain outliers as their runtime is rather ``deterministic'', in the sense that their runtime is accurately represented as a polynomial function of the number of columns of an instance.
On the contrary, the runtimes of TSPConcorde contain a multitude of outliers, while most runtimes are still below 100ms.
Only two instances take more than 10 seconds to solve.

\subsection{Constraints}
Next, we  present experiments on how constraints on the column order affect the runtime and the number of blocks of TSPConcorde.
Namely, we implemented the constraints from \cref{sec:prioritiessets,sec:constraintsweight} that either specify that two sets/rows have to be represented as a single line segment/consecutive blocks of ones, or give specific weights to sets. The evaluation for both constraints 
works as follows.
\begin{itemize}
  \item Two sets as single line segment: We pick uniformly at random for each instance in our test set $T_1\cup T_2$ two sets that have to be drawn as a single line segments, and then apply the reduction to TSP described in \cref{sec:prioritiessets}, and solve the resulting TSP-instance with the Concorde TSP-solver. We identify this approach by TSPConcordeFS for ``\textbf{f}ixed \textbf{s}ets''.
  \item Weighted sets: For each matrix $A$ in the test set $T_1\cup T_2$ we specify a weight function $f:[m_A]\to\N$  %
  that assigns to each set a unique integer weight in $[m_A]$ uniformly at random. Then, we apply the reduction to TSP as described in \cref{sec:constraintsweight} and solve the resulting TSP-instance with the Concorde TSP-solver. We identify this approach by TSPConcordeW for ``\textbf{w}eighted''.
\end{itemize}

\begin{table}[!t]
  \centering
  \caption{Results for test set $T_1$ and constrained versions. White columns depict mean relative/absolute optimality gaps (except TSPConcorde); gray columns depict mean runtimes.}
  \label{table:expt1constr}
  \begin{tabularx}{\linewidth}{l >{\centering\arraybackslash}X >{\centering\arraybackslash\columncolor{lightgray!15}}X >{\centering\arraybackslash}X >{\centering\arraybackslash\columncolor{lightgray!15}}X >{\centering\arraybackslash}X >{\centering\arraybackslash\columncolor{lightgray!15}}X}
    \toprule
    \multirow{4}*{\#columns} & \multicolumn{2}{c}{TSPConcorde} & \multicolumn{2}{c}{TSPConcordeFS} & \multicolumn{2}{c}{TSPConcordeW}                             \\
                             & blocks /                        & $t$                               & gap                              & $t$  & gap         & $t$  \\
                             & row                             & [ms]                              & [rel./abs.]                      & [ms] & [rel./abs.] & [ms] \\\midrule
    10                       & 1.7                             & 7                                 & 4.1/1.7                          & 8    & 2.6/0.9     & 7    \\
    20                       & 2.8                             & 13                                & 4.8/3.7                          & 17   & 4.5/3.0     & 11   \\
    30                       & 4.0                             & 22                                & 6.2/6.7                          & 37   & 6.5/5.9     & 25   \\
    50                       & 6.0                             & 69                                & 6.4/10.2                         & 104  & 8.7/12.1    & 61   \\
    70                       & 7.9                             & 340                               & 7.0/14.8                         & 210  & 9.4/16.7    & 84   \\
    \bottomrule
  \end{tabularx}
\end{table}
\begin{table}[!t]
  \centering
  \caption{Results for test set $T_2$ and constrained versions. White columns depict the mean relative and absolute optimality gaps (except TSPConcorde); gray columns depict the mean runtimes.}
  \label{table:expt2const}
  \begin{tabularx}{\linewidth}{l >{\centering\arraybackslash}X >{\centering\arraybackslash\columncolor{lightgray!15}}X >{\centering\arraybackslash}X >{\centering\arraybackslash\columncolor{lightgray!15}}X >{\centering\arraybackslash}X >{\centering\arraybackslash\columncolor{lightgray!15}}X}
    \toprule
    \multirow{3}*{\#columns} & \multicolumn{2}{c}{TSPConcorde} & \multicolumn{2}{c}{TSPConcordeFS} & \multicolumn{2}{c}{TSPConcordeW}                             \\
                             & blocks /                        & $t$                               & gap                              & $t$  & gap         & $t$  \\
                             & row                             & [ms]                              & [rel./abs.]                      & [ms] & [rel./abs.] & [ms] \\\midrule
    20-50                    & 1.7                             & 17                                & 3.6/0.9                          & 28   & 4.0/0.9     & 16   \\
    55-80                    & 1.8                             & 23                                & 3.1/0.8                          & 36   & 4.8/1.2     & 27   \\
    85-110                   & 1.9                             & 36                                & 3.0/0.8                          & 57   & 5.2/1.4     & 30   \\
    115-140                  & 2.0                             & 82                                & 3.2/0.9                          & 79   & 6.0/1.7     & 40   \\
    145-160                  & 2.0                             & 71                                & 3.1/0.8                          & 99   & 6.2/1.7     & 38   \\
    \bottomrule
  \end{tabularx}
\end{table}

\begin{figure}[!t]
  \centering
  \includegraphics[width=.71\textwidth]{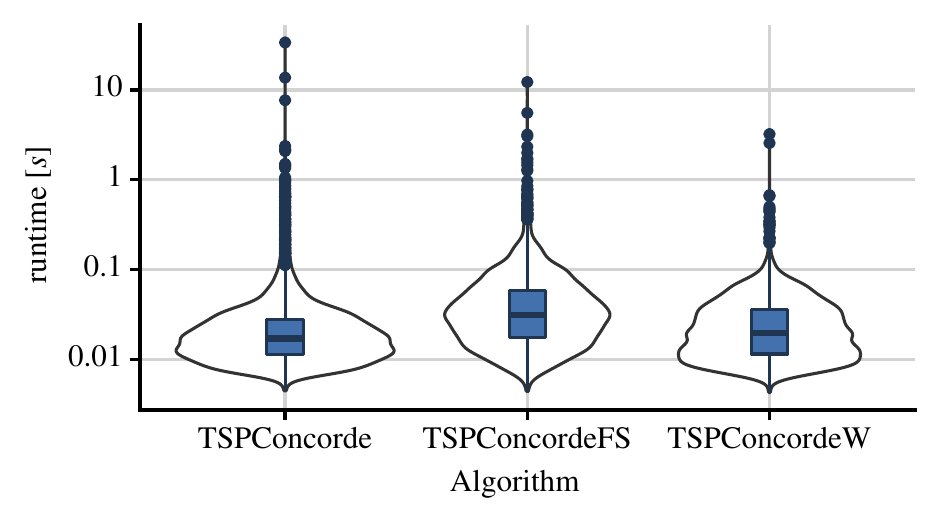}
  \caption{Runtimes of constrained algorithms for all instances from $T_1$ and $T_2$.}
  \label{fig:runtimesconstr}
\end{figure}

\cref{table:expt1constr} shows runtimes and optimality gaps for test set $T_1$. We observe that adding constraints does not influence runtimes of the TSP solver significantly. Furthermore, by adding constraints we may not be able to reach the optimal number of line segments anymore and see a maximum optimality gap of 10\%. The results for test set $T_2$ are similar and are given in \cref{table:expt2const}.

\cref{fig:runtimesconstr} shows a box and violin plot of runtimes for TSPConcorde and the constrained versions thereof, further suggesting that adding constraints does not significantly influence runtimes.

\section{Conclusion}
We have studied the algorithmic complexity of computing optimal linear diagrams and observed that it is equivalent to a related problem on binary matrices. 
Despite its \NP-completeness, even in a restricted setting, we have formulated a TSP model for solving the problem optimally. In an experimental study, we have seen that a state-of-the-art TSP solver can in fact solve a large set of instances obtained from our model optimally, most of them within few milliseconds. 
Hence it is feasible to strive for optimal linear diagrams in most practical settings and thus reduce the number of line segments by up to 10\% compared to the best heuristics, which, otherwise, are faster by one to two orders of magnitude. 
\bibliographystyle{splncs04}
\bibliography{literaturenonupdated.bib}

\end{document}